%% file: main.tex
\documentclass{article}
\usepackage{kr}

\usepackage{times}
\usepackage{soul}
\usepackage{url}
\usepackage[hidelinks]{hyperref}
\usepackage[utf8]{inputenc}
\usepackage[small]{caption}
\usepackage{graphicx}
\usepackage{amsmath}
\usepackage{amsthm}
\usepackage{booktabs}
\usepackage{algorithm}
\usepackage{algorithmic}
\usepackage{latexsym}
\usepackage{thmtools}
\usepackage{thm-restate}
\usepackage{todonotes}
\usepackage{xspace}
\usepackage{amssymb}
\usepackage{stmaryrd}
\usepackage{mathtools}
\usepackage{nameref}
\usepackage{enumitem}
\usepackage{subfig}
\usepackage{tikz}
\usepackage{relsize}
\usepackage{textgreek}

\newtheorem{theorem}{Theorem}
\newtheorem{proposition}{Proposition}
\newtheorem{observation}{Observation}
\newtheorem{lemma}{Lemma}

\makeatletter 
\newcommand{\mylabel}[2]{#2\def\@currentlabel{#2}\label{#1}}
\makeatother

\include{macros}

\title{On the Relationship between Shy and Warded Datalog+/-}

\author{%
Teodoro Baldazzi$^1$\and
Luigi Bellomarini$^2$\and
Marco Favorito$^{2,3}$\and
Emanuel Sallinger$^{4,5}$ \\
\affiliations
$^1$Universit\`a Roma Tre, Italy \\
$^2$Bank of Italy \\
$^3$Universit\`a Sapienza, Italy \\
$^4$TU Wien, Austria \\
$^5$University of Oxford, UK \\
\emails
\begin{large}
teodoro.baldazzi@uniroma3.it,
\{luigi.bellomarini,marco.favorito\}@bancaditalia.it,
sallinger@dbai.tuwien.ac.at
\end{large}
}

\begin{document}

\maketitle

\begin{abstract}
Datalog$^\exists$ is the extension of Datalog with existential quantification.
While its high expressive power, underpinned by a simple syntax and the support for full recursion, renders it particularly suitable for modern applications on knowledge graphs, query answering (QA) over such language is known to be undecidable in general.
For this reason, different fragments have emerged, introducing syntactic limitations to Datalog$^\exists$ that strike a balance between its expressive power and the computational complexity of QA, to achieve decidability.
In this short paper, we focus on two promising tractable candidates, namely Shy and Warded Datalog$^\pm$.
Reacting to an explicit interest from the community, we shed light on the relationship between these fragments.
Moreover, we carry out an experimental analysis of the systems implementing Shy and Warded, respectively DLV$^\exists$ and Vadalog.
\end{abstract}

\input{body}

\newpage

\section*{Acknowledgements}
The work on this paper was supported by the Vienna Science and Technology Fund (WWTF) grant VRG18-013


\bibliographystyle{kr}
\typeout{}
\bibliography{biblio}

\end{document}

%% file: macros.tex
\newcommand{\Nat}{{\rm I\kern-.23em N}}

\def\diamondplus{\mathrel{%
  \ooalign{\raise.21ex\hbox{$\tiny+$}\cr\hss$\Diamond$\hss}}}
\def\diamondminus{\mathrel{%
  \ooalign{\raise.21ex\hbox{$\tiny-$}\cr\hss$\Diamond$\hss}}}

\newcommand{\vadalog}{\textsc{Vadalog}\xspace}
\newcommand{\dlv}{\textsc{DLV}$^\exists$\xspace}

\newcommand{\shy}{\ensuremath{\mathsf{shy}}\xspace}
\newcommand{\warded}{\ensuremath{\mathsf{warded}}\xspace}
\newcommand{\protectedd}{\ensuremath{\mathsf{protected}}\xspace}

\newcommand{\chase}{\ensuremath{chase}}
\newcommand{\ochase}{\ensuremath{ochase}}
\newcommand{\ichase}{\ensuremath{ichase}}
\newcommand{\pchase}{\ensuremath{pchase}}
\newcommand{\pchaser}{\ensuremath{pchase_r}}


%% file: body.tex

\section{Introduction}
\label{introduction}
The last decade has witnessed a rising interest, both in academia and industry, towards querying and exploiting data in the form of \textit{knowledge graphs} (KGs), modeled by combining extensional knowledge with ontological theories to infer intensional information.
This led to the adoption of novel intelligent systems that perform \textit{ontological reasoning} and \textit{ontology-based query answering} (QA) tasks over KGs, employing powerful logic languages for knowledge representation~\cite{KrTh16}.

Among the main requirements such languages must exhibit, a high expressive power is essential in modern applications on KGs, so as to model and reason on complex domains with full recursion and existential quantification~\cite{bellomarini2018swift}.
At the same time, decidability and tractability of QA must be sustained, limiting the data complexity to a polynomial degree~\cite{gottlob2015beyond}.

In this context, Datalog$^\exists$, the natural extension of Datalog with existential quantification in rule heads, became particularly relevant.
Its semantics is specified in an operational way via the \textsc{chase}~\cite{MaMS79}, an algorithmic tool that takes as input a database $D$ and a set $\Sigma$ of rules, and modifies $D$ by adding new tuples until $\Sigma$ is satisfied.
While this language encompasses both a high expressiveness and a simple syntax that enable powerful knowledge-modeling, QA over it is known to be undecidable in general~\cite{cali2013taming}.
For this reason, recent years have witnessed a rise of proposals for decidable classes of Datalog$^\exists$ in the literature~\cite{cali2012general,cali2012towards,baget2010walking,baget2011walking}, defined by imposing proper syntactic limitations to strike a good balance between the expressive power of the language and the computational complexity of QA.

In this paper, we focus on two particularly promising languages, namely, \textit{Shy Datalog$^\exists$}~\cite{leone2019fast} and \textit{Warded Datalog$^\pm$}~\cite{gottlob2015beyond}, which have been independently introduced.
Indeed, they both cover the requirements for knowledge representation, restricting Datalog$^\exists$, though with notably different constraints, and featuring $\textsf{PTIME}$ data complexity for Boolean conjunctive QA (BCQA).
Shy and Warded are employed in state-of-the-art reasoning systems.
Specifically, Shy was introduced as part of the \textit{parsimonious} class~\cite{leone2012efficiently} and is adopted in the system \textit{DLV$^\exists$}~\cite{dlvE}.
Likewise, Warded was recently introduced as fragment of the Datalog$^\pm$ family~\cite{cali2010datalog+} and is implemented as logic core of the reasoner \textit{Vadalog system}~\cite{bellomarini2020vadalog}.
Both find many industrial applications in the financial, media intelligence, security, logistics, pricing domains, and more~\cite{BGPS19,AACC18}.

Previous works have thoroughly discussed how these two languages individually compare to the other main decidable classes~\cite{leone2019fast,gottlob2014datalog+}.
However, to our surprise, determining the relationship between Shy and Warded is still an unexplored research topic, despite known to be recurring in the Datalog$^\exists$ academic and practitioner communities.

This paper aims at providing an answer to such question, by contributing the results summarized below.
After a brief overview of Shy and Warded, we present a novel \textbf{theoretical analysis} of the relationship between the languages.
From a syntactical point of view, we conclude that they intersect in a newly defined fragment, named \textbf{Protected Datalog$^\pm$}.
Regarding semantics, we show that the \textsc{chase} procedures adopted by Shy and Warded are equivalent over Protected settings with respect to BCQA.
To enrich the analysis with a more empirical perspective, we then illustrate an \textbf{experimental comparison} between \dlv and the Vadalog system on QA tasks over Protected settings.

\smallskip \noindent\textbf{Overview.} This paper is organized as follows.
In Section~\ref{sec:background}, we provide an overview of Shy and Warded.
In Section~\ref{sec:relationship}, we discuss the relationship between them.
In Section~\ref{sec:experiments}, we provide the experimental evaluation of \dlv and the Vadalog system.
We draw our conclusions in Section~\ref{sec:conclusion}.


\section{Shy Datalog$^\exists$ and Warded Datalog$^\pm$}
\label{sec:background}
To guide our discussion, we briefly recall some relevant notions and provide an overview of the two fragments at issue.
Let $C$, $N$, and $V$ be disjoint countably infinite sets of \textit{constants}, (\textit{labelled}) \textit{nulls} and (\textit{regular}) \textit{variables}, respectively.
A (\textit{relational}) \textit{schema} $\mathbf{S}$ is a finite set of relation symbols (or \textit{predicates}) with associated arity.
A \textit{term} is either a constant or variable.
An \textit{atom} over $\mathbf{S}$ is an expression of the form $R(\bar v)$, where $R \in \mathbf{S}$ is of arity $n > 0$ and $\bar v$ is an $n$-tuple of terms. 
A \textit{database} (\textit{instance}) over $\mathbf{S}$ associates to each relation symbol in $\mathbf{S}$ a relation of the respective arity over the domain of constants and nulls. The members of the relations are called \textit{tuples} or \textit{facts}.
Given two conjunctions of atoms \textvarsigma$_1$ and \textvarsigma$_2$, we define a \textit{homomorphism} from \textvarsigma$_1$ to \textvarsigma$_2$ as a mapping $h: C \cup N \cup V \rightarrow C \cup N \cup V$ such that $h(t)=t$ if $t \in C$, $h(t) \in C \cup N$ if $t \in N$ and if $a(t_1, \ldots, t_n)$ is an atom $\in$ \textvarsigma$_1$, then $a(h(t_1), \ldots, h(t_n)) \in$ \textvarsigma$_2$: \textvarsigma$_1$ and \textvarsigma$_2$ are \textit{isomorphic} if $h^{-1}$ is a homomorphism from \textvarsigma$_2$ to \textvarsigma$_1$.

A Datalog$^\exists$ program consists of a set of facts and \emph{existential rules} $\forall \bar x \forall \bar y (\varphi(\bar x,\bar y)$$\to$$\exists \bar z~\psi(\bar x, \bar z))$, where $\varphi$ (the \textit{body}) and $\psi$ (the \textit{head}) are conjunctions of atoms.
We omit $\forall$ and denote conjunction by comma.
Let $\Sigma$ be a set of Datalog$^\exists$ rules and $p[i]$ a position (i.e., the $i$-th term of a predicate $p$ with arity $k$, where $i=1,\ldots,k$).
We define $p[i]$ as \textit{affected} if (i)~$p$ appears in a rule in $\Sigma$ with an existentially quantified variable (\textit{$\exists$-variable}) in $i$-th term  or, (ii)~there is a rule in $\Sigma$ such that a universally quantified variable (\textit{$\forall$-variable}) is only in affected body positions and in $p[i]$ in the head.

\smallskip\noindent\textbf{Shy Datalog$^\exists$.}
Let $y$ be an $\exists$-variable in $\Sigma$. We define the position $p[i]$ as \textit{invaded} by $y$ if there is a rule $\rho$ $\in$ $\Sigma$ such that \textit{head}($\rho$) $=$ $p$($t_1,\ldots,t_k$) and either (i)~$t_i=y$ or, (ii)~$t_i$ is a $\forall$-variable that occurs in \textit{body($\rho$)} only in positions that are invaded by $y$.
By such definition, if $p[i]$ is invaded, then it is affected, but not vice versa.
Let $x$ $\in$ $\textbf{X}$ be a variable in a conjunction of atoms \textvarsigma$_{[\textbf{X}]}$.
We say that $x$ is \textit{attacked} in \textvarsigma$_{[\textbf{X}]}$ by $y$ if $x$ occurs in \textvarsigma$_{[\textbf{X}]}$ only in positions invaded by $y$.
If $x$ is not attacked by any variable, $x$ is \textit{protected} in \textvarsigma$_{[\textbf{X}]}$.

We define a set $\Sigma$ as \textit{Shy Datalog$^\exists$} (or $\shy$) if, for each rule $\sigma$ $\in$ $\Sigma$: \mylabel{cond:s1}{(S1)}~if a variable $x$ occurs in more than one body atom, then $x$ is protected in \textit{body}($\sigma$); and, \mylabel{cond:s2}{(S2)}~if two distinct $\forall$-variables are not protected in \textit{body}($\sigma$) but occur both in \textit{head}($\sigma$) and in two different body atoms, then they are not attacked by the same variable~\cite{leone2019fast}.

\smallskip\noindent\textbf{Warded Datalog$^\pm$.}
A $\forall$-variable $x$ is \textit{harmful}, wrt a rule $\rho$ in $\Sigma$, if $x$ appears only in affected positions in $\rho$, otherwise it is \textit{harmless}.
A (join) rule that contains a harmful (join) variable is a \textit{harmful} (\textit{join}) \textit{rule}. If the harmful variable is in \textit{head($\rho$)}, it is \textit{dangerous}~\cite{gottlob2015beyond}.

We define a set $\Sigma$ as \textit{Warded Datalog$^\pm$} (or $\warded$) if, for each rule $\sigma$ $\in$ $\Sigma$: \mylabel{cond:w1}{(W1)}~all the dangerous variables appear in a single body atom, called \textit{ward}; and, \mylabel{cond:w2}{(W2)}~the ward only shares harmless variables with other atoms in the body.

\noindent\textbf{Chase and Semantics.}
Chase-based procedures enforce the satisfaction of $\Sigma$ over a database $D$ ($\big\langle D$,$\Sigma\big\rangle$), incrementally expanding $D$ into new instances $I$ with facts derived from the application of the rules in $\big\langle D$,$\Sigma\big\rangle$, until $\Sigma$ is satisfied ($\chase(D,\Sigma)$).
Such facts may contain labelled nulls as placeholders for the $\exists$-variables~\cite{cali2010datalog+}.
We refer to its \textit{oblivious} variant with $\ochase$.
Consider an instance $I^\prime \supseteq I$.
Given a rule $\sigma: \varphi(\bar x,\bar y)$$\to$$\exists \bar z~\psi(\bar x, \bar z))$ $\in$ $\Sigma$, a \textit{chase step} $\big\langle\sigma$,$h\big\rangle$ is \textit{applicable} to $I^\prime$ if there exists a homomorphism \textit{h} that maps the atoms of $\varphi(\bar x,\bar y)$ to facts of $I$ (i.e., $h$($\varphi(\bar x,\bar y)$) $\subseteq$ $I$).
When the chase step is applicable, the atom $h^\prime(\psi(\bar x, \bar z))$ is added to $I^\prime$, where $h^\prime$ is obtained by extending $h$ so that $h^\prime(z_i) \in N$ is a fresh labelled null, $\forall$ $z_i \in \bar z$.

However, in the presence of recursion, especially jointly with existential quantification, infinite labelled nulls could be generated in $\ochase$, causing the procedure not to terminate and inhibiting the decidability of the QA task~\cite{CaGL09}.
To practically achieve termination and decidability, $\shy$ and $\warded$ both employ variants of the $\ochase$, based on \textit{firing conditions} that limit the applicability of the chase steps.
Specifically, $\shy$ adopts the so-called \textit{parsimonious} chase ($\pchase$): the chase step $\big\langle\sigma$,$h\big\rangle$ is applicable wrt $I^\prime \supseteq I$ if, additionally, there is no homomorphism from \textit{h(head($\sigma$))} to $I^\prime$.
To cover decidability of CQ cases and preserve correctness of the evaluation, $\pchase$ is extended into its variant \textit{with resumption} ($\pchaser$)~\cite{leone2019fast}, which iteratively ``resumes'' the chase in the same state it was after termination but considering previous nulls as constants.
Similarly, $\warded$ employs an \textit{isomorphism-based} chase ($\ichase$): the chase step $\big\langle\sigma$,$h\big\rangle$ is applicable wrt $I^\prime \supseteq I$ if, additionally, there is no isomorphic embedding of \textit{h(head($\sigma$))} to $I^\prime$.
Here, decidability of CQA derives from~\cite[Theorem 1, Theorem 2]{bellomarini2018vadalog}, as we shall see.

A \textit{Boolean Conjunctive Query} (BCQ) is a first-order expression $q: \exists \textbf{Y}$~\textvarsigma$_{[\textbf{X}\cup\textbf{Y}]}$, where $\textbf{X} \in C$.
The answer of $q$ over an instance $I$ (namely, BCQA) is \textit{true}, denoted by $I\models q$, if and only if there is a homomorphism \textit{h}: $\textbf{Y}$ $\rightarrow$ $C \cup N$ s.t.\ \textit{h}(\textvarsigma$_{[\textbf{X}\cup\textbf{Y}]}$) $\subseteq$ $I$.
It holds that $q$ is true over $\chase(D,\Sigma)$, denoted by $\chase(D,\Sigma)\models q$, if and only if $\big\langle D$,$\Sigma\big\rangle\models q$.
We recall that the query output tuple problem (decision version of CQ evaluation) and BCQ evaluation are AC$_0$-reducible to each other~\cite{cali2012general}.
Therefore, we only consider BCQ evaluation.


\section{Relationship between Shy and Warded}
\label{sec:relationship}
We now have the means to discuss the relationship between $\shy$ and $\warded$.
Let us start by showing that the two fragments are not characterized by any form of syntactical containment (represented with the symbol $\not\subseteq$).

\begin{proposition}
\label{prop:shy-not-sub-warded}
$\shy \not\subseteq \warded$
\end{proposition}
\begin{proof}
We prove the claim by showing that there exists a program that is \shy but not \warded. Let $\Sigma$ be the following:
\begin{align*}
    \textit{E$_1$}(\text{x})\to\exists{\text{y}}~\textit{I$_1$}(\text{x},\text{y}) \tag{$\alpha$}\\
    \textit{E$_2$}(\text{x})\to\exists{\text{z}}~\textit{I$_2$}(\text{z},\text{x}) \tag{$\beta$}\\
    \textit{I$_1$}(\text{x},\text{y}),\textit{I$_2$}(\text{z},\text{x})\to\textit{I$_3$}(\text{x},\text{y},\text{z}) \tag{$\rho$}
\end{align*}

\noindent
Here, rules~$\alpha$ and~$\beta$ are existential rules, both $\shy$ and $\warded$, and introduce affectedness in positions $I_1[2]$ and $I_2[1]$, respectively. Rule~$\rho$ has a harmless join on $x$ and propagates the harmful variables $y$ and $z$ to the head. Indeed, $\rho$ is not $\warded$, as $y$ and $z$ are dangerous and there is no ward (condition~\ref{cond:w1} is not satisfied). However, it is $\shy$, as the join variable $x$ is protected (condition~\ref{cond:s1}), whereas $y$ and $z$ are attacked by distinct variables, respectively $y_\alpha$ and $z_\beta$ (condition~\ref{cond:s2}). Therefore $\Sigma$ is $\shy$ but it is not $\warded$.
\end{proof}

\begin{proposition}
\label{prop:warded-not-sub-shy}
$\warded \not\subseteq \shy$
\end{proposition}
\begin{proof}
We prove the claim by showing that there exists a program that is $\warded$ but not $\shy$. Let $\Sigma$ be the following:
\begin{align*}
    \textit{E$_1$}(\text{x})\to\exists{\text{y}}~\textit{I$_1$}(\text{x},\text{y}) \tag{$\alpha$}\\
    \textit{I$_1$}(\text{x},\text{y}),\textit{I$_1$}(\text{z},\text{y})\to\textit{I$_2$}(\text{x},\text{z}) \tag{$\rho$}
\end{align*}
Here, rule~$\alpha$ is an existential rule, both $\warded$ and $\shy$, and introduces affectedness in position $I_1[2]$. Rule~$\rho$ contains a harmful join on $y$ and propagates the harmless variables $x$ and $z$ to the head. Indeed, $\rho$ is $\warded$, as no dangerous variables occur in the rule. However, it is not $\shy$, as the join variable $y$ is attacked by $y_\alpha$ (condition~\ref{cond:s1} is not satisfied). Therefore $\Sigma$ is $\warded$ but it is not $\shy$.
\end{proof}

\smallskip\noindent\textbf{Protected Datalog$^\pm$.}
To further explore the relationship between $\shy$ and $\warded$, we first introduce the notion of \textit{protected harmful} variable. Given a Datalog$^\exists$ set $\Sigma$, a $\forall$-variable $x$ is protected harmful, with respect to a rule $\rho$ $\in$ $\Sigma$, if it appears in affected positions in $\rho$ that are not invaded by the same $\exists$-variable: if the invading variable is the same, $x$ is an \textit{attacked harmful} variable. Without loss of generality (as more complex joins can be broken into steps~\cite{bellomarini2020vadalog}), we define \textit{protected harmful join rule} as a rule:
\begin{equation}
\label{eq:protected-harmful-join}
    A(x_1,y_1,h),B(x_2,y_2,h) \to \exists{z}~C(\overline{x},z) \tag{$\tau$}
\end{equation}
where $A$, $B$ and $C$ are atoms, $A[3]$ and $B[3]$ are positions invaded (and thus affected) by distinct $\exists$-variables, $x_1, x_2 \subseteq \overline{x}$, $y_1, y_2 \subseteq \overline{y}$ are disjoint tuples of harmless variables or constants and $h$ is a protected harmful variable.
By definition of protected variables and labelled nulls, the join on $h$ only activates on constant values in the \textsc{chase}. If $h$ is attacked, $\tau$ is an \textit{attacked harmful join rule}.

We define a set $\Sigma$ as \textit{Protected Datalog$^\pm$} (or $\protectedd$) if, for each rule $\sigma$ $\in$ $\Sigma$: \mylabel{cond:p1}{(P1)}~$\sigma$ does not contain attacked harmful joins; and, \mylabel{cond:p2}{(P2)}~$\sigma$ is $\warded$ (it satisfies~\ref{cond:w1} and~\ref{cond:w2}).
With reference to the relationship between $\shy$ and $\warded$, we first show that $\protectedd$ corresponds to the syntactical intersection of the two fragments (represented with $\cap$).

\begin{theorem}
\label{th:warded-and-shy-eq-protected-warded}
\warded $\cap$ \shy $=$ \protectedd.
\end{theorem}
\begin{proof}
We prove the equivalence by showing the containment in both directions of implication.

$\warded \cap \shy \subseteq \protectedd$. Let $\Sigma$ be a generic set of rules $\in$ $\warded \cap \shy$. Then, $\Sigma$ satisfies condition~\ref{cond:p2} by definition. Regarding condition~\ref{cond:p1}, we proceed by contrapositive and show that $\lnot\ref{cond:p1}\implies\lnot\ref{cond:s1}$. Indeed, if $\Sigma$ does not satisfy condition~\ref{cond:p1}, then there exists a rule $\sigma$ $\in$ $\Sigma$ with an attacked harmful join. However, this means that $\sigma$ contains a variable that occurs in more than one body atom and it is not protected in \textit{body}($\sigma$) (condition $\ref{cond:s1}$ is not satisfied). Therefore, by contrapositive, $\Sigma$ $\in$ $\protectedd$.

$\protectedd \subseteq \warded \cap \shy$. Let $\Sigma$ be a generic set of rules $\in$ $\protectedd$. By condition~\ref{cond:p2}, $\Sigma$ $\in$ $\warded$. Also, by condition~\ref{cond:p1}, $\Sigma$ may only contain harmless joins and protected harmful joins, thus an attacked variable cannot occur in both body atoms by definition. Therefore, condition~\ref{cond:s1} is satisfied. Finally, we proceed by contrapositive and show that $\lnot\ref{cond:s2}\implies\lnot\ref{cond:p2}$. Indeed, if $\Sigma$ does not satisfy condition~\ref{cond:s2}, then there exists a rule $\sigma$ $\in$ $\Sigma$ with two dangerous attacked variables in distinct body atoms. However, this means that $\sigma$ does not contain a ward, thus it is not $\warded$ (condition $\ref{cond:p2}$ is not satisfied). Therefore, by contrapositive, $\Sigma$ $\in$ $\warded \cap \shy$. This concludes the proof.
\end{proof}

\noindent
Figure~\ref{fig:languages} illustrates the syntactic containment of these fragments, as well as their data complexity.

\smallskip
With reference to the semantic perspective of this analysis, let us first develop the following consideration.

\begin{observation}
$\pchase(D,\Sigma)\subseteq \ichase(D,\Sigma)$, $\forall\Sigma\in$ Datalog$^\exists$, $\forall D$. This derives from the definition of $\pchase$ and $\ichase$, as the applicability of their chase steps depends on fact homomorphism and fact isomorphism, respectively. In particular, whenever a $\pchase$ step $\langle \sigma,h \rangle$ is applicable with respect to $I^\prime \supseteq I$, the absence of homomorphisms from $h(\textit{head}(\sigma))$ to $I^\prime$ implies the absence of isomorphic embeddings of $h(\textit{head}(\sigma))$ to $I^\prime$ (and not vice versa). The $\ichase$ step is therefore applicable as well.
\end{observation}

\noindent
Since, by Theorem~\ref{th:warded-and-shy-eq-protected-warded}, $\protectedd$ is a syntactical subset of $\warded$, its data complexity is also $\textsf{PTIME}$, as shown in Figure~\ref{fig:languages}.
Moreover, we observe that reasoning with $\protectedd$ can adopt both $\pchaser$ and $\ichase$.
Based on these notions, we make an additional step in the comparative analysis of the fragments, stating that $\pchaser$ and $\ichase$ are equivalent over $\protectedd$ settings, with respect to BCQA (i.e., a generic BCQ has the same answer).

\begin{lemma}
\label{lem:ochase-eq-pchaser-shy}
Let $\Sigma\in\shy$, $D$ a database and $q$ a BCQ. Then, $\ochase(D,\Sigma)\models q$ if and only if $\pchaser(D,\Sigma)\models q$.
\end{lemma}
\begin{proof}
The result directly follows from~\cite[Theorem 4.9]{leone2019fast} for BCQA decidability over \shy.
\end{proof}

\begin{lemma}
\label{lem:ochase-eq-ichase-warded}
Let $\Sigma\in\warded$, $D$ a database and $q$ a BCQ. Then, $\ochase(D,\Sigma)\models q$ if and only if $\ichase(D,\Sigma)\models q$.
\end{lemma}
\begin{proof}
Completeness, namely $\ochase(D,\Sigma)\models q$ implies $\ichase(D,\Sigma)\models q$, directly follows from~\cite[Theorem 1, Theorem 2]{bellomarini2018vadalog}, as chase subgraphs derived from isomorphic facts are isomorphic, thus irrelevant for BCQA. Soundness, namely $\ichase(D,\Sigma)\models q$ implies $\ochase(D,\Sigma)\models q$, holds because $\ichase(\Sigma)\subseteq\ochase(\Sigma)$ by definition. 
\end{proof}

\begin{theorem}
\label{th:pchaser-ichase-bcq-eq-protected}
Let $\Sigma\in\protectedd$, $D$ a database and $q$ a BCQ. Then, $\pchaser(D,\Sigma)\models q$ if and only if $\ichase(D,\Sigma)\models q$.
\end{theorem}
\begin{proof}
By Theorem~\ref{th:warded-and-shy-eq-protected-warded} and by definition of $\protectedd$ fragment, we know that $\Sigma\in\shy$ and $\Sigma\in\warded$. Therefore, the result directly follows from Lemma~\ref{lem:ochase-eq-pchaser-shy} and Lemma~\ref{lem:ochase-eq-ichase-warded}.
\end{proof}

\begin{figure}[hbt!]
  \centering
  \includegraphics[width=0.20\textwidth]{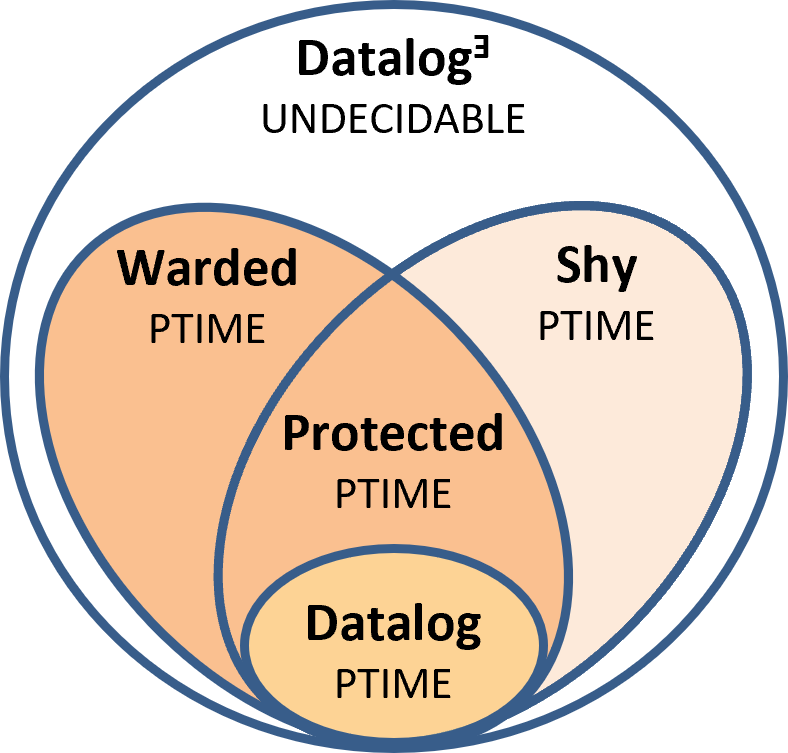}
  \caption{Syntactic containment of fragments and data complexity.}
  \label{fig:languages}
\end{figure}


\section{DLV$^\exists$ and Vadalog System}
\label{sec:experiments}
We integrate the analysis of $\shy$ and $\warded$ with an experimental comparison between their main state-of-the-art implementations, namely \dlv and the Vadalog system ($\vadalog$). Indeed, it follows from Section~\ref{sec:relationship} that the two systems are comparable over $\protectedd$ settings.

\smallskip\noindent\textbf{Systems Overview.}
\dlv is an extension of the answer set programming system DLV~\cite{leone2006dlv}, enriched with $\pchaser$ for CQA over $\shy$ programs.
To answer CQs, it employs a \textit{materialization} approach, producing and storing all the facts for each predicate via the so-called \textit{semi-naive} evaluation~\cite{abiteboul1995foundations}, where rules are evaluated according to a bottom-up strategy from the initial database.
It is available online~\cite{dlvE}.
$\vadalog$ is a well-known system for KG management, implementing $\warded$ and $\ichase$ for reasoning and CQA~\cite{bellomarini2018vadalog}.
To answer CQs, it employs a \textit{streaming} approach, building a \textit{reasoning query graph} as a processing pipeline, where nodes correspond to algebra operators that perform transformations over the data pulled from their predecessors, and edges are dependency connections between the rules.
It is available upon request.
While \dlv integrates powerful optimization techniques that \vadalog has yet to incorporate, the latter is also extended with multiple features of practical utility, such as aggregations and equality-generating dependencies.

\smallskip\noindent\textbf{Experiments and Results.}
We compared \dlv and $\vadalog$ over distinct reasoning scenarios and QA tasks. The experiments were run on a local installation of the two systems, using a machine equipped with an Intel Core i7-8665U CPU running at 1.90 GHz and 16 GB of RAM. The results of the experiments, as well as the steps to reproduce them on \dlv, were made public~\cite{experiments}.

The first set of experiments is based on a financial recursive scenario about persons and companies~\cite{bellomarini2018vadalog} and real datasets extracted from DBPedia~\cite{dbpedia}. A \textit{person of significant control} (PSC) for a company is a person that directly or indirectly has some control over the company. The goal of this task is finding all the PSCs for the companies in DBPedia. We ran it for all the 67K available companies and for 1K, 10K, 100K, 500K and 1M of the available persons. Figure~\ref{fig:psc} illustrates similar execution times for the two systems, all under 5 seconds. Specifically, \dlv has better times in the first cases, partially due to $\vadalog$'s longer pre-processing phase for the creation of the query graph. With larger datasets, $\vadalog$'s performance progressively improves, thanks to its efficient recursion control to avoid the exploration of redundant areas of the reasoning space, and its \textit{routing strategies} to traverse the query graph~\cite{bellomarini2020vadalog}.

\begin{figure}[hbt!]
  \centering
  \includegraphics[width=0.32\textwidth]{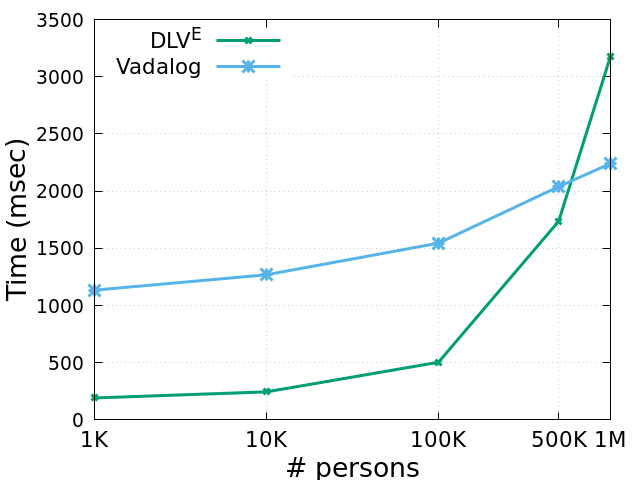}
  \caption{QA times for DBPedia PSC scenarios.}
  \label{fig:psc}
\end{figure}

The second set of experiments is based on \textit{Doctors}, a data integration task from the schema mapping literature~\cite{mecca2014iq}, included in the \textsc{ChaseBench} benchmark~\cite{benedikt2017benchmarking}. It represents a plausible real-world case related to the healthcare domain and it features existential quantification. We ran it for 10K, 100K, 500K and 1M of all the datasets and over 7 distinct queries, of which we report the average times.
Figure~\ref{fig:doctors} illustrates that, while both systems show very good behaviour even in the most demanding cases, \dlv outperforms $\vadalog$. This is motivated by the powerful optimization techniques integrated in \dlv that limit the loading of redundant data for the query and reduce the space needed for materializing the output of $\pchaser$~\cite{leone2019fast}.

\begin{figure}[hbt!]
  \centering
  \includegraphics[width=0.32\textwidth]{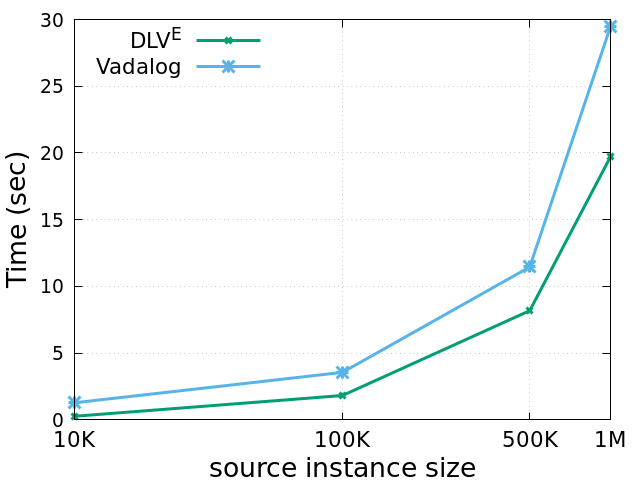}
  \caption{QA times for \textsc{ChaseBench} Doctors scenarios.}
  \label{fig:doctors}
\end{figure}


\section{Conclusion}
\label{sec:conclusion}
Shy Datalog$^\exists$ and Warded Datalog$^\pm$ are two relevant languages that extend Datalog with existential quantification while sustaining decidability of BCQA. Reacting to an explicit interest of the community, we provided an analysis of the fragments in terms of syntactical relationship and query evaluation, as well as an experimental comparison of their main implementations. Future work includes investigating their mutual reduction into the intersection fragment we defined and its impact in terms of their semantic relationship.
